%% file: sample-aamas19.tex
\documentclass[sigconf]{aamas}  
\usepackage{balance}  

\usepackage{booktabs}

\setcopyright{ifaamas}  
\acmDOI{}  
\acmISBN{}  
\acmConference[AAMAS'19]{Proc.\@ of the 18th International Conference on Autonomous Agents and Multiagent Systems (AAMAS 2019)}{May 13--17, 2019}{Montreal, Canada}{N.~Agmon, M.~E.~Taylor, E.~Elkind, M.~Veloso (eds.)}  
\acmYear{2019}  
\copyrightyear{2019}  
\acmPrice{}  

\usepackage{booktabs}
\usepackage{amssymb}
\usepackage{graphicx}
\usepackage{algpseudocode}
\usepackage{algorithm}
\usepackage{enumitem}
\usepackage{subcaption}
\usepackage{gensymb}
\usepackage{multirow}
\usepackage[most]{tcolorbox}

\begin{document}

\title{Minimizing Travel in the Uniform Dispersal Problem for Robotic Sensors}


\author{Michael Amir}
\email{ammicha3@cs.technion.ac.il}
\affiliation{%
  \institution{Technion - Israel Institute of Technology}
  \city{Haifa} 
  \state{Israel} 
}

\author{Alfred M. Bruckstein}
\email{freddy@cs.technion.ac.il}
\affiliation{%
  \institution{Technion - Israel Institute of Technology}
  \city{Haifa} 
  \state{Israel}
}

\begin{abstract}  
The limited energy capacity of individual robotic agents in a swarm often limits the possible cooperative tasks they can perform. In this work, we investigate the problem of covering an unknown connected grid environment (e.g. a maze or connected corridors) with a robotic swarm so as to minimize the maximal number of steps that each member of the swarm makes and their activity time before their work is finished, thereby minimizing the energy requirements. The robots are autonomous, anonymous and identical, with local sensors and finite memory, and possess no communication capabilities. They are assumed to disperse over time from a fixed location, and to move synchronously. The robots are tasked with occupying every cell of the environment, while avoiding collisions.

In the literature such topics are known as \textit{uniform dispersal problems}. The goal of minimizing the number of steps traveled by the robots has previously been studied in this context. Our contribution is a local robotic strategy for simply connected grid environments that, by exploiting their topology, achieves optimal makespan (the amount of time it takes to cover the environment) and minimizes the maximal number of steps taken by the individual robots before their deactivation. The robots succeed in discovering optimal paths to their eventual destinations, and finish the covering process in $2V-1$ time steps, where $V$ is the number of cells in the environment. 
\end{abstract}

%


\maketitle


\input{samplebody-conf}


\newpage 
\bibliographystyle{ACM-Reference-Format}  
\balance  
\bibliography{sample-bibliography}  

\end{document}

%% file: samplebody-conf.tex
\setlength{\belowcaptionskip}{-11pt}
\section{Introduction and related work}

The objective of swarm robotics is to enable a large group of simple and autonomous mobile robots to work cooperatively towards complex goals. It is often the case, e.g. when the robots are traveling large distances or are airborne, that a lot of energy is required for the sustained activity of robots in the swarm. In this work, we are interested in solving the \textit{uniform dispersal} problem for simply connected grid environments while minimizing the movement and active time of each individual robot, in order to minimize the energy requirements.

In many real life scenarios, e.g. mapping or hazard detection, one is interested in deploying agents over an unknown area and covering it for the purposes of sensing or reacting \cite{howard2002}. The use of swarm robotics to solve such problems has many inherent advantages, such as scalability, greater coverage, and autonomy in mission execution. In \textit{uniform dispersal}, a large number of mobile robots emerge over time from a source or several source locations (called ``doors'' in the literature), and are tasked to completely cover an unknown environment $R$ by occupying every location and to terminate their work in finite time \cite{barraswarm1}. The robots must not collide (i.e. two robots must never occupy the same location), nor step outside the boundaries of the environment. 

Hsiang et al. \cite{hsiang} \cite{hsiang2003geom} introduced the problem of uniform dispersal in discrete planar domains by mobile robots endowed only with finite memory, local sensors, and local communication. Their DFS-esque ``follow the leader'' strategy enables robots to cover the environment in optimal time, assuming a synchronous time scheme. Much follow-up work has focused on achieving dispersal with weaker models of robots, e.g. disallowing communication, reducing memory, or assuming asynchronous time \cite{barraswarm1} \cite{hideg2017uniformtime} \cite{barrameda2013uniform}. Barrameda et al. \cite{barraswarm1} have shown that the dispersal problem is intractable under the usual assumptions if the robots are assumed to be oblivious (that is, to possess no persistent states), though there have been attempts to get around this limitation using randomization \cite{hideg2017uniform}. It is standard to assume that the robots are moving in a connected grid environment, as any 2D space can be approximated well by pixelation into tiny grid cells of uniform size.

From a theoretical perspective, the problem of dispersing and coordinating mobile robotic agents while minimizing movement or energy has been studied extensively both as a centralized motion planning problem and in distributed sensor networks \cite{demaine2009minimizing} \cite{minimizing2} \cite{minimizing3}, and various computational hardness results have been proven in the case of general graph environments \cite{minimizing1}. More broadly, multi-agent scheduling problems have been studied in the presence of energy constraints \cite{hartuv2018scheduling}. Specifically in the context of uniform dispersal for robotic sensors, the question of minimizing travel for orthogonal areas that we here concern ourselves with was discussed in the original paper by Hsiang et. al \cite{hsiang} and soon after in Stainzberg's doctoral dissertation \cite{sztainberg2003algorithms}, and more recently in \cite{hideg2016area} and  \cite{hideg2017uniform}. 

In recent decades there has been considerable effort dedicated to the algorithmic problems of agent coverage or exploration, wherein a robot or team of robots must completely explore, occupy, or map an area. Attention has been given to the case of a single robot tasked with visiting every vertex of a graph or grid environment \cite{batalin2007design} \cite{yehoshua2015frontier}, to single- and multi-robot path planning \cite{agmon2008giving}, to natural or pheromone-based computation models \cite{payton2001pheromone} \cite{wagner1997line} \cite{wagner2000mac}, to related formation or dispersal problems \cite{morlok2007dispersing} \cite{cheng2008coalition}, and to a multitude of other topics. We refer the reader to \cite{galceran2013survey} or \cite{altshuler2018introduction} for recent surveys. The problem of uniform dispersal distinguishes itself from many of these by its distinctly online nature. The robots emerge onto the environment at different times and must successfully embed themselves into the ongoing exploration effort, without colliding with other robots, and without interrupting the constant outflow of new robots. They must do this under stringent computational, sensory, and communication restrictions--in most recent models, the robots, modelled as finite automata, are not allowed to talk to each other, and cannot even tell the difference between environmental obstacles and the presence of robots active in the formation. We find it fairly surprising that under these restrictions, robots are capable of exploring an entirely unknown environment in theoretically optimal time, as well as (we shall see) walk only in shortest paths to their destinations while doing so.

Much attention has been given to the problem of deployment and coverage in GPS-denied environments, as this may enable the deployment of robotic fleets outside laboratory conditions and their utilization in real world scenarios. Dispersal strategies that operate under stringent restrictions on communication and sensing may be especially relevant to future investigations in this domain. Implementation, however, forms a technical barrier, as when looking at the problem of generating robust uniform coverage from a systems perspective the issues of relative visual localization - range, angular coverage and persistence - become important. There has been progress towards overcoming these barriers in a number of different settings. In \cite{biswas2012depth} the authors discuss a visual relative localization method suited for autonomous navigation and obstacle avoidance in indoor environments, for mobile robots with limited computational power. In \cite{saska2017system} the authors present a visual localization method based on an image processing algorithm suitable for use on small quadcopters. The algorithm assumes all the quadcopters have identical but specific markings that ease the localization. These are but examples of the sensors an agent might use when implementing strategies that operate under such restrictions.

\textbf{Our contribution:} Working in a synchronous time setting, Hsiang et al. \cite{hsiang} pose the problem of minimizing the total and individual number of steps the robots take (the ``total travel'' and ``individual travel''), while achieving optimal makespan--the time before complete coverage of the environment. They describe several algorithms for \textit{general} grid environments that consecutively improve on each other in this respect, but these algorithms do not achieve a global optimum. 

We describe a local uniform dispersal strategy that, for \textit{simply connected} grid environments, achieves optimal makespan and minimizes the total travel and maximal individual travel. The strategy's goal is to enable a robot to settle in place as soon as possible, thereby minimizing the energy consumption. It exploits the ability to decompose simply connected environments into a tree of simply connected sub-environments via  ``halls''--defined as corners of the environment that also have an obstacle located diagonally opposite to them. We work in a setting similar to \cite{hsiang}, where time is synchronous and robots have local sensors and finite memory. Specifically, the robots require $5$ bits of persistent memory ($2^5$ persistent states), and a visibility span of Manhattan distance 2.  As is sometimes assumed, e.g. in \cite{hideg2017uniformtime}, they are initialized with a common notion of up, down, left and right. Unlike \cite{hsiang}, our algorithm works without assuming any inter-robot communication capabilities: the robots are only capable of seeing environmental obstacles (including other robots that block them), and are unable to distinguish between kinds of obstacles.

By attempting to restrict their movement to as few directions as possible, our strategy enables the robots to travel in shortest paths from their arrival point to their eventual, a-priori unknown, settling point. The robots finish dispersing in $2V-1$ time steps, where $V$ is the number of cells in the environment. 

We show further that no local strategy can minimize total travel in the general case, i.e. for general grid environments.

\section{Model}
Consider the integer grid  $\mathbb{Z}^2 = \mathbb{Z} \times \mathbb{Z}$, whose vertices are points $(x,y)$ where $x$ and $y$ are both integers, and $(x_1, y_1)$ is connected to $(x_2, y_2)$ if and only if the Manhattan distance $|x_1 - x_2| + |y_1 - y_2|$ is exactly $1$. A grid environment or region $R$ is defined as a connected sub-graph of $\mathbb{Z}^2$. The \textit{complement} of $R$, denoted $R^c$, is defined as the sub-graph $\mathbb{Z}^2 - R$ of $\mathbb{Z}^2$. We call the vertices of $R^c$ \textit{walls}. 

\begin{definition}
\label{holelessdefinition}
A region $R$ is said to be  \textbf{simply connected} if and only if any path $v_1 v_2 \ldots v_1$ of vertices in $R$ that forms a closed curve does not surround any vertices of $R^c$.
\end{definition}

In particular, a region $R$ is simply connected if $R^c$ is connected.

A robot is a mobile point in $R$ with limited vision and small finite memory. No two robots may occupy the same location. The visibility range of all robots is assumed to be $2$, meaning at every time step, a robot is aware of unoccupied vertices in $R$ that are at a Manhattan distance of 2 or less from it. It infers from this the positions of local \textit{obstacles} (walls or other robots), but cannot distinguish between types of obstacles. All robots have a shared notion of up, down, left and right upon emergence from $s$. 

Time is discretized to steps of $t = 1, 2, \ldots$. At every time step, all robots perform a Look-Compute-Move operation sequence, in which they examine their environment and move to a new location based on a computation they perform (a robot may also choose to stay in place - this counts as a move). This occurs synchronously, meaning that all robots move to their computed next location at the same time. The ``beginning'' of a time step refers to the configuration of the robots at that time step before the robots move. The ``end'' of a time step is the configuration at that time step after the robots move.

We denote by $prev(A)$ the position of a robot $A$ at the beginning of the previous time step, and by $next(A)$ its position at the beginning of the next time step.

A given robot is either \textit{active} or \textit{settled}. All robots are initially active, and eventually become settled at the end of some time step. Settled robots never move from their current position.

A unique vertex $s$ in $R$ is designated as the source or ``door'' vertex. If at the beginning of a time step there is no mobile robot at $s$, a new robot emerges at $s$ at the end of that time step. 

\textit{Energy and total travel.} The ``travel'' $T_i$ of the $i$th robot is the number of time steps $t$ that begin and end with the robot still active. This definition includes steps where the robot does not change location, since we wish to relate travel to energy expenditure (e.g., a quadcopter floating or circling in place is still traveling, and consumes just as much energy). The total travel of the robots is then the sum $\sum T_i$ over all robots, and can be seen as the total amount of energy the robots consume before they settle.

\section{Find-Corner Depth-First Search}

We describe a local rule, ``Find-Corner Depth-First Search'' (Algorithm \ref{alg:FCDFS}), that enables the robots to disperse over a simply-connected region $R$. As in  \cite{hsiang}, the algorithm has a makespan of $2V-1$ (where $V$ is the number of cells in $R$, or equivalently, the total area of $R$ when setting every cell to be a unit square). We note that since at best, robots arrive at $s$ once per two time steps, this is the lowest possible makespan.

The purpose of FCDFS is to minimize the  individual travel and total travel of the robots. It does this by ensuring that the path of a robot from $s$ to its eventual destination (the vertex at which it settles) is a shortest path in $R$. 

The idea of the algorithm lies in the distinction between a \textit{corner} and a \textit{hall} (see Figure \ref{fig:corners} and Figure \ref{fig:halls}):

\begin{definition}
A vertex $v$ of a grid environment $R$ is called a \textbf{corner} if either:
 
\begin{enumerate}[label=(\alph*)]
    \item $v$ has one or zero neighbours in $R$, or
    \item $v$ has precisely two neighbours $u$ and $u'$ in $R$, and $u$ and $u'$ have a common neighbour $w$ that is distinct from $v$.
\end{enumerate}
\end{definition}

\begin{definition}
A vertex $v$ of $R$ is called a \textbf{hall} if it has precisely two neighbours $u$ and $u'$, and $u$ and $u'$ are both adjacent to the same vertex $w$ in $R^c$.
\end{definition}

\begin{figure}[!htb]
  \centering%
    \includegraphics[width=.2\linewidth,scale=0.3]{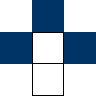}\hfill%
    \includegraphics[width=.2\linewidth,scale=0.3]{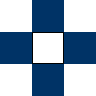}\hfill%
    \includegraphics[width=.2\linewidth,scale=0.3]{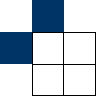}\hfill%
    \caption{Corners. (Blue vertices are \textit{walls}; vertices in $R^c$).}
    \label{fig:corners}
\end{figure}

\begin{figure}[!htb]
  \centering%
    \includegraphics[width=.2\linewidth]{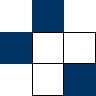}\hfill%
    \caption{A hall.}
    \label{fig:halls}
\end{figure}

Essentially, halls are vertices in $R$ that are blocked by walls on two sides, and have an additional wall $w$ diagonal to them. Corners are either dead-ends, or vertices in $R$ that are blocked by walls on two sides, and have a vertex $w$ of $R$ diagonal to them. If $v$ is either a hall or a corner, $w$ is called the ``diagonal'' of $v$, and is denoted $diag(v)$. We observe that diagonals are uniquely specified.

Robots executing FCDFS attempt to move only in `primary' and `secondary' directions, where the secondary direction is always a 90-degree clockwise rotation of the primary direction (for example "up and right", "right and down", or "down and left"). They may only change their primary direction once they arrive at a hall, and they become settled once both their primary and secondary directions are blocked and they are at a corner.

For the rest of this section, let $R(t)$ be the environment $R$ at time $t$, i.e. the initial environment $R$ where we have removed from $R$ every vertex that is occupied by a \textit{settled} robot at the beginning of time step $t$. 

A robot at time $t$ is searching for the corners and halls of $R(t)$. However, robots executing FCDFS are unable to distinguish between active robots, and walls or settled robots. Hence, it is important to design the algorithm so that a robot never misidentifies a corner of $R(t)$ as a hall, or vice-versa, due to an active robot (rather than a wall or a  settled robot) occupying the diagonal and being identified as an obstacle. For this purpose we enable our robots to remember their two previous locations. We will show that an active robot can occupy the diagonal of a  corner at time $t$ if and only if its predecessor occupied this diagonal at time $t-2$, thereby allowing the predecessor to distinguish between 'real' and 'fake' halls.

\begin{algorithm}[!htb]
  \caption{Find-Corner Depth-First Search}
  \begin{algorithmic}
    \State Let $v$ be the current location of $A$.
    \If{every neighbouring vertex of $v$ is occupied}
        \State Settle.
    \ElsIf{$A$ has never moved}\Comment{Initialization}
        \State Search clockwise, starting from the "up" direction, for an unoccupied vertex, and set primary direction to point to that vertex.
    \EndIf
    \If{$A$ can move in its primary direction}
        \State Step in the primary direction.
    \ElsIf{$A$ can step in secondary direction}
        \State Step in the secondary direction.
    \Else\Comment{We are at a corner or a hall.}
    \If{$prev(prev(A)) = diag(v)$ $\lor$ $diag(v)$ is unoccupied} 
        \State Settle.
    \Else\Comment{We think we are at a hall.}
        \State Set primary direction to point to the neighbour of $v$ different from $prev(A)$.
        \State Move in the primary direction.
    \EndIf
    \EndIf
  \end{algorithmic}
  \label{alg:FCDFS}
\end{algorithm}

\subsection{Analysis}

In this section we give an analysis of the FCDFS algorithm. To start, we require some lemmas about corners and halls.

\begin{lemma}
\label{removecornersimplyconnected}
Let $c$ be a corner of a simply connected region $R$. Then:

\begin{enumerate}[label=(\alph*)]
\item $R - c$ is simply connected.
\item For any two vertices $u, v$ in $R - c$, the distance between $u$ and $v$ is the same as in $R$.
\end{enumerate}
\end{lemma}

\begin{proof}
Removing $c$ does not affect connectedness, nor does it affect the distance from $u$ to $v$, as any path going through $c$ can instead go through $diag(c)$. Further, as $c$ is adjacent to two walls, no path in $R-c$ can surround it, so $R-c$ also remains simply connected.   
\end{proof}

An \textit{articulation point} (also known as a separation or cut vertex) is a vertex of a graph whose deletion increases the number of connected components of the graph (i.e. disconnects the graph) \cite{reinharddiestel2017}.

\begin{lemma}
\label{treestructure}
The halls of a simply connected region are articulation points.
\end{lemma}

\begin{proof}
Let $h$ be a hall of a simply connected region $R$. Suppose for contradiction that $h$ is not an articulation point, and let $u$ and $u'$ be the neighbours of $h$. Then there is a path from $u$ to $u'$ that does not pass through $h$. Let $P$ be this path, and let $P'$ be the path from $u$ to $u'$ that goes through $h$. 

When embedded in the plane in the usual way, $R$ is in particular a simply connected topological space. The hall $h$ is embedded onto a unit square, whose four corners each touch a wall: three touch the two walls adjacent to $h$, and the fourth touches $diag(h)$. Joined together to form a closed curve, the paths $P$ and $P'$ form a rectilinear polygon that must contain at least one corner of $h$ in its interior. Hence, the curve $PP'$ contains a part of $R^c$--and we get a contradiction to the simply connected assumption. (See Figure \ref{fig:hall_lemma}). 

\begin{figure}[!ht]
  \centering%
    \includegraphics[width=.49\linewidth]{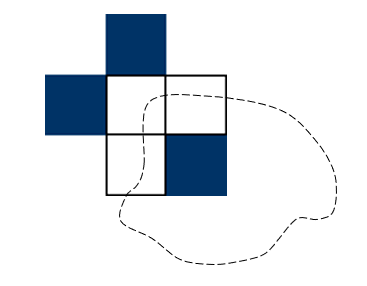}\hfill%
    \includegraphics[width=.49\linewidth]{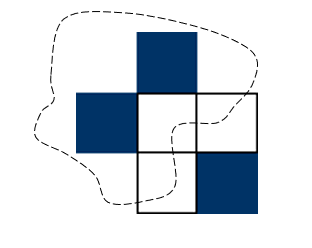}\hfill%
    \caption{The two possibilities for $PP'$.}
    \label{fig:hall_lemma}
\end{figure}

\end{proof}

Lemma \ref{treestructure} indicates that $R$ can be decomposed into a tree structure $T(R)$ as follows: first, delete all halls of $R$ to form separated connected components. Let $C_1, C_2, \ldots, C_n$ be these components, where $C_i$ also includes its adjacent halls. Letting the vertices of $T(R)$ be these components, connect $C_i$ and $C_j$ by an edge if they share a hall. We set $C_1$ to be the root of the tree, and the connected component containing the door vertex $s$.

By Lemma \ref{removecornersimplyconnected}, assuming our robots correctly stop only at corners, $R(t)$ can in the same manner be decomposed into a tree $T(R(t))$ whose connected components are $C_1(t), C_2(t), \ldots$. These components are each a sub-graph of a connected component of $T(R)$. 

Let $A_1, A_2, \ldots$ denote the robots that emerge from $s$ in the order of arrival. In the next several propositions, we make the \textit{no fake halls at time $t$} assumption: this is the assumption that for any $t' < t$, at the end of time step $t'$: robots can only become settled at corners of $R(t')$, and can only change primary directions at halls of $R(t')$. We do not include the initialization of a primary direction when a robot arrives at $s$. We will later show that the ``no fake halls'' assumption is always true, so the propositions below hold unconditionally.

\begin{proposition}
\label{shortestpathproposition}
Assuming no fake halls at time $t$, a robot $A_i$ active at the beginning of time step $t$ has traveled an optimal path in $R$ from $s$ to its current position.
\end{proposition}

\begin{proof}
By the assumption, the only robots that became settled did so at corners. Consequently, by Lemma \ref{removecornersimplyconnected}, $R(t)$ is a connected graph, and there is a path in $R(t)$ from $s$ to $A_i$. The path $A_i$ took might not be in $R(t)$, but whatever articulation points (and in particular halls) $A_i$ passed through must still exist, by definition.

Since $A_i$ is active at the beginning of time $t$, by the algorithm, it has taken a step every unit of time up to $t$. Until $A_i$ enters its first hall, and between any two halls $A_i$ passes through, it only moves in its primary and secondary directions. This implies that the path $A_i$ takes between the halls of $R(t)$ must be optimal (since it is optimal when embedded onto the integer grid $\mathbb{Z}^2$). We note also that $A_i$ never returns to a hall $h$ it entered a connected component of $R(t)$ from, since the (possibly updated) primary direction pulls it away from $h$.

We conclude that $A_i$'s path consists of taking locally optimal paths to traverse the connected components of the tree $T(R(t))$ in order of increasing depth. Since in a tree there is only one path between the root and any vertex, this implies that $A_i$'s path to its current location is at least as good as the optimal path in $R(t)$. By Lemma \ref{removecornersimplyconnected}, b, this implies that $A_i$'s path is optimal in $R$.  
\end{proof}

\begin{corollary}
\label{distancecorollary}
Assuming no fake halls at time $t$,
\begin{enumerate}[label=(\alph*)]
\item For all $i < j$, the distance between the robots $A_i$ and $A_j$, if they are both active at the beginning of $t$, is at least $2(j-i)$
\item No collisions (two robots occupying the same vertex) have occurred. 
\end{enumerate}
\end{corollary}

\begin{proof}
For proof of (a), note that at least two units of time pass between every arrival of a new robot (since in the first time step after its arrival, a newly-arrived robot blocks $s$). Hence, when $A_j$ arrives, $A_i$ will have walked an optimal path towards its eventual location at time $t$, and it will be at a distance of $2(j-i)$ from $s$. This distance is never shortened up to time $t$, as $A_i$ will keep taking a shortest path. 

(b) follows immediately from (a).  
\end{proof}

From Corollary \ref{distancecorollary} and determinism, we get:

\begin{lemma}
\label{followtheleaderlemma}
Suppose $A_i$ is active at the beginning of time step $t$. Assuming no fake halls at time $t$, $next(A_{i+1}) = prev(A_i)$.
\end{lemma}

We note that Lemma \ref{followtheleaderlemma} also indicates that if at the beginning of time step $t$, $A_i$ is active, then $A_{i+1}$ will be active at the beginning of time step $t+1$. 

We can now show that the ``no fake halls'' assumption is true, and consequently, the propositions above hold unconditionally.

\begin{proposition}
\label{nofakehalls}
For any $t$, at the end of time step $t$: robots only become settled at corners of $R(t)$, and only change primary directions halls of $R(t)$ (not including the primary direction decided at initialization).
\end{proposition}

\begin{proof}
The proof of the proposition is by induction. The base case for $t=1$ is trivially true.

Suppose that up to time $t-1$, the proposition holds. Note that this means the ``no fake halls'' assumption holds up to time $t$, so we can apply the lemmas and propositions above to the algorithm's configuration at the beginning of time $t$. 

We will show that the proposition statement also holds at time $t$. Let $A_i$ be an active robot whose location at the beginning of $t$ is $v$. First, consider the case where $v = s$. The algorithm only enables $A_i$ to settle at $s$ if it is surrounded by obstacles at all directions. Any obstacle adjacent to $A_i$ must be a wall of $R(t)$ (as any active robot must be at a distance at least $2$ from $A_i$, due to Corollary \ref{distancecorollary}). Hence, if $A_i$ settles at $s$, $s$ is necessarily a corner, as claimed.

We now assume that $v \neq s$. We separate the proof into two cases:

Case 1: Suppose $A_i$ becomes settled at the end of time step $t$. Then by the algorithm, at the beginning of $t$, $A_i$ detects obstacles in its primary and secondary directions. These must be walls of $R(t)$ due to Corollary \ref{distancecorollary}, so $v$ is either a corner or a hall of $R(t)$. Since $A_i$ settled, we further know that either $diag(v)$ is empty, or $prev(prev(A_i)) = diag(v)$. In the former case, $v$ is a corner of $R(t)$. In the latter case, we know from Lemma \ref{followtheleaderlemma} and from the fact that no collisions occur that the only obstacle detected at $diag(v)$ is $A_{i+1}$, which is an active robot, so $v$ is again a corner of $R(t)$. In either case a corner is detected and the agent is settled.

Case 2: Suppose $A_i$ changed directions at the end of time step $t$. Then it sees two adjacent obstacles, and an obstacle at $diag(v)$. As in case 1, we infer that $v$ is either a corner or a hall. If it is a corner, then $diag(v)$ is an active agent. By Corollary \ref{distancecorollary}, it is either $A_{i+1}$ or $A_{i-1}$. It cannot be $A_{i+1}$, as then $A_i$'s position two time steps ago would have been $diag(v)$, so it would become settled instead of changing directions. It cannot be $A_{i-1}$, as $diag(v)$ is closer to $s$ than $v$, and $A_{i-1}$ has arrived earlier than $A_i$, and has been taking a shortest path to its destination. Hence, $diag(v)$ cannot be an active agent, and $v$ must be a hall as claimed.  
\end{proof}

We have shown that the no fake-hall assumption is justified at all times $t$, hence we can assume that the propositions introduced in this section hold unconditionally.

\begin{proposition}
\label{timeanalysis}
Let $V$ be the number of vertices of $R$. At the end of time-step $2V-1$, every cell is occupied by a robot.
\end{proposition}

\begin{proof}
Propositions \ref{shortestpathproposition} and \ref{nofakehalls} imply that robots take a shortest path in $R$ to their destination. That means that as long as the destination of a robot is not $s$ itself, robots will step away from $s$ one unit of time after they arrive. Until then, this means that robots arrive at $s$ at rate one per two time steps.

Every robot's end-destination is a corner, and by the initialization phase of the algorithm, the destination is never $s$ unless $s$ is completely surrounded. Since there are no collisions, there can be at most $V$ robots in $R$ at any given time. By Lemma \ref{removecornersimplyconnected}, robots that stop at corners keep $R$ connected. Furthermore, every $R(t)$ is a rectilinear polygon, so unless it has exactly one vertex, it necessarily has at least two corners. This means that the destination of every robot is different from $s$ unless $s$ is the only unoccupied vertex. Hence, a robot whose destination is $s$ will only arrive when $s$ is the only unoccupied vertex, and this will happen when $V$ robots have arrived, so after at most $2V-1$ time steps. This is exact, since it is impossible to do better than $2V-1$.  
\end{proof}

Propositions \ref{timeanalysis} and \ref{shortestpathproposition}, alongside the ``no fake halls'' proof, complete our analysis. They show that FCDFS has a makespan of $2V-1$, and also that the durations of activity of the individual robots are optimal, since every robot travels a shortest path to its destination without stopping. 

As every vertex must be occupied for the dispersal to end, a trivial lower bound on the total travel for any dispersal algorithm is $\sum_{v \in R}{dist(s,v)}$. Since this is achieved by our algorithm, total travel is also minimized.

In practice, the energy savings of our algorithm are dependent on the shape of the environment $R$. We take as a point of comparison the Depth-First Leader-Follower algorithm of Hsiang et al. \cite{hsiang}. On a 1-dimensional line of length $n$, both FCDFS and DFLF require the same total travel, $O(n^2)$, so no improvement is attained. In contrast, on an $n$-by-$n$ square grid, DFLF requires total travel $O(n^4)$ in the worst case, and FCDFS requires $O(n^3)$ - significantly less. This is because the DFLF strategy starting from a corner might cause the leader, $A_1$, to ``spiral'' inwards into the grid, covering every one of its $n^2$ vertices in $n^2 - 1$ moves; the subsequent robot $A_i$ will make $n^2 - i$ moves, for a sum total of $O(n^4)$. FCDFS, on the other hand, distributes the path lengths more uniformly. Note that both algorithms take the exact same amount of time to finish.

\textit{Where is it best to place $s$?} If we want to minimize the total travel, by the formula given above, the best place to place $s$ is the vertex of $R$ that minimizes the sum of distances $\sum_{v \in R}{dist(s,v)}$ (there may be several). This is the discrete analogue of the so-called Fermat-Toricelli point, or the ``geometric median'' \cite{krarup1997torricelli}. 

\subsection{The number of persistent states}\label{section:numstates} 

\begin{algorithm}[ht]
  \caption{5-bit FCDFS}
  \begin{algorithmic}
    \State Let $v$ be the current location of $A$.
    
    \If{$v$ has no unoccupied neighbours}
        \State Settle.
        \State $b_3b_4b_5 \gets 011$
    \ElsIf{$b_4b_5 = 00$}
        \State Search clockwise, starting from the "up" direction, for an unoccupied vertex, and set primary direction to point to that vertex.
        \State $b_4b_5 \gets 10$
    \EndIf
    
    \If{$A$ cannot move in primary or secondary directions}
        \If {$v$ has just one neighbour}
            \State Settle.
            \State $b_3b_4b_5 \gets 011$
        \ElsIf{$(b_5 = 1 \land b_3+b_4=1)$ $\lor$ $diag(v)$ is unoccupied} 
            \State Settle.
            \State $b_3b_4b_5 \gets 011$
        \Else  
            \State Set primary direction to obstacle-less direction not equal to $180\degree$ rotation of previous direction stepped in (i.e. the neighbour of $v$ we haven't visited yet; this can be inferred from $b_1b_2$ and $b_3$).  
            \State $b_4b_5 \gets 10$
        \EndIf
    \EndIf
    
    \If{$b_4b_5$ was not updated at this time step}\Comment{i.e. $b_5 = 1$ or time to update $b_5$}
        \State $b_4b_5 \gets b_3 1$
    \EndIf
    
    \If{$A$ can move in its primary direction}
        \State Step in the primary direction.
        \State $b_3 \gets 0$
    \ElsIf{$A$ can step in secondary direction}
        \State Step in the secondary direction.
        \State $b_3 \gets 1$
    \Else
        \State Settle.
        \State $b_3b_4b_5 \gets 011$
    \EndIf
  \end{algorithmic}
  \label{alg:5bitFCDFS}
\end{algorithm}

As in previous work on uniform dispersal, our robots are finite-state automatons with $O(1)$ persistent memory bits or states that carry over between time steps. The requirement of finite memory is important, as it allows for scalability: the robots' memory need not scale with the size or complexity of the environment. 

There has been some interest in the question of just how little memory one can get away with. It has been shown that oblivious robots - robots with just one persistent state - are incapable of solving the dispersal problem, even with infinite visibility \cite{barraswarm1}. Consequently, any dispersal algorithm requires some number of persistent states, and we are interested in implementing our algorithm with as few as possible - i.e. bringing the robots as close as possible to ``obliviousness'' of their prior history and to center their decisions, as much as possible, on their current position and frame of reference. 

Moreover, Algorithm \ref{alg:FCDFS} required the robots to remember their previous locations relative to their current location and to be able to use them as points of comparison. The 5-bit implementation shows how this could be done through remembering only the previous two relative directions of motion. A robot is then required only to know whether there are obstacles at the four cardinal directions (up, down, left, right), and at its diagonal, which is always at a $135\degree$ degree rotation from the primary direction. This simplifies the localization computations. 

We implemented a 5-bit or $2^5$-state version of our algorithm on a simulator (see Algorithm \ref{alg:5bitFCDFS}). A robot's state is described by bits $b_1b_2b_3b_4b_5$. All bits are initially $0$. $b_1b_2$ describe the primary direction (one of four), and $b_3$ tells us whether the previous step was taken in the primary direction (if $b_3 = 0$) or in the secondary direction (if $b_3 = 1$). $b_4b_5$ is a counter that is reset to $10$ upon entering a hall or one step after initialization, and thereafter is equal to $*1$, where $*$ is a bit that tells us whether we walked in the primary or secondary direction two steps ago (by copying $b_3$). A robot that detects an obstacle at its diagonal interprets its position as a fake hall (i.e. a corner) as long as $b_5 = 1$ and $b_3 + b_4 = 1$, that is, as long as at least one time step passed since the last hall, and our previous position was diagonal to us. In order to conserve memory, our robots do not strictly speaking have a ``settled'' state. Instead, once a robot determines it is in a corner (and so needs to settle), it sets $b_3b_4b_5$ to $011$, indicating that it visited its diagonal--this causes it to never move again.

\subsection{The impossibility of minimizing total travel for general grid environments} 

We saw that there is a local rule that minimizes total travel for simply connected grid environments. In this section we show that, for robots with finite visibility, there is no local rule that universally minimizes total travel for all connected grid environments.

Let $r$ be the visibility range of the robots. Consider the grid environment in Figure \ref{fig:impossibility} (not drawn to scale). It connects a set of $10r$ columns of width 1 spaced $2r$ cells apart. The bottom row has total length $20r^2$. Most of the columns are dead-ends and have a height of $30r^2$. The first column and an additional column connect to the top row, and have height $30r^2 + 1$. Label the grid environment where this additional column is the $k$th column $G(k)$. The door $s$ is at the bottom left.
 
\begin{figure}[!ht]
    \centering
    \includegraphics[height=2.086in]{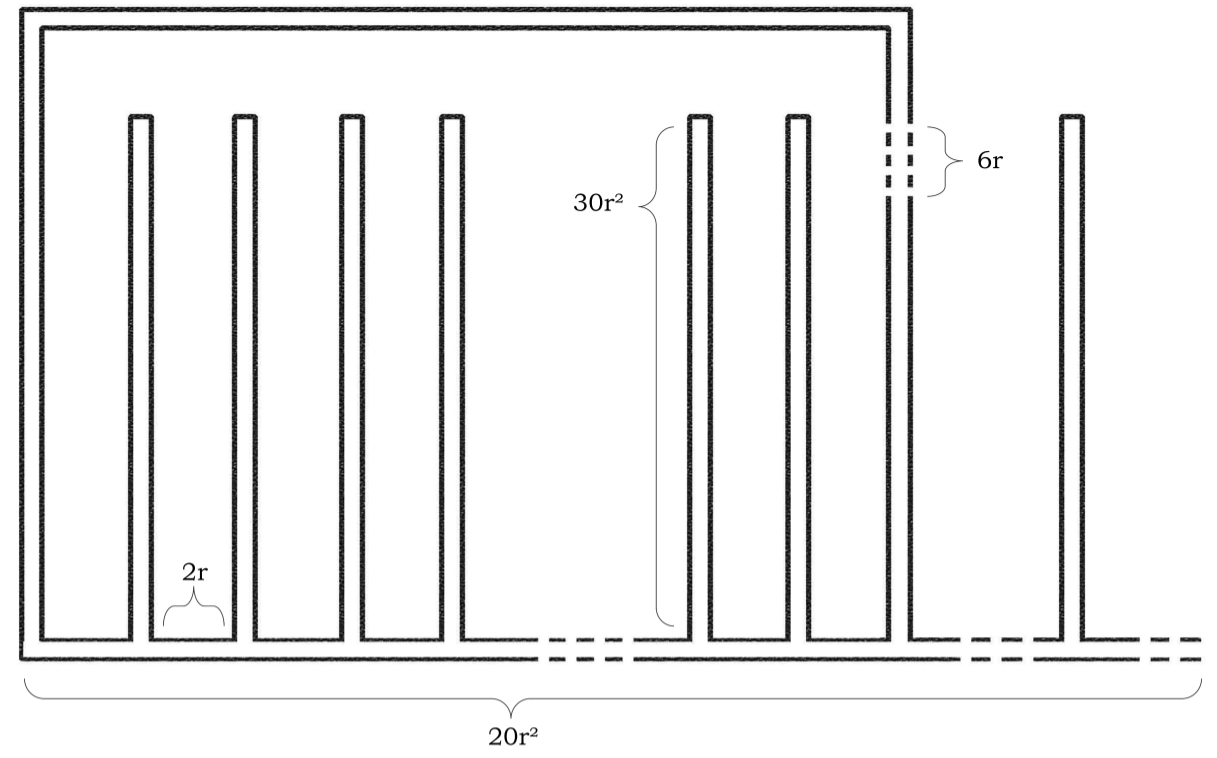}
    \caption{The construction $G(k)$.}
    \label{fig:impossibility}
\end{figure}

It is readily seen that the total travel required by an optimal solution for any environment $G(k)$ is $\sum_{v \in G(k)}{dist(s,v)}$, where $s$ is the door of $G(k)$ (let a line of robots going up the first column fill the top row, and let robots going to the right fill the other columns).

\begin{proposition}
\label{impossibility}
Let \textbf{ALG} be a local rule for uniform dispersal of robots with visibility range $r$. There is an environment $G(k)$ for which the total travel of \textbf{ALG} is at least $\sum_{v \in R}{dist(s,v)} + 1$.
\end{proposition}

\begin{proof}
\textit{(Sketch)} We consider the actions of rule $\textbf{ALG}$ on the grid environment $G(k)$. We do not specify the value of $k$ yet.

As before, label the robots emerging at $s$ $A_1, A_2, \ldots$ in their order of arrival. Since $A_1$ cannot distinguish between the up and right directions upon arrival at $s$ (any distinct feature of the environment is at distance at least $r+1$ and hence is invisible), we can assume without loss of generality that it steps up (if it steps right, simply rotate and reflect $G(k)$). 

Assume for contradiction that the total travel of \textbf{ALG} is $T = \sum_{v \in G(k)}{dist(s,v)}$. This assumption implies that every robot travels a shortest path to its settlement destination. In particular, $A_1$ must have precisely $dist(A_1, v_1)$ travel, where $v_1$ is the destination at which $A_1$ chooses to settle. 

We note the following facts:

\begin{enumerate}
    \item Once $A_1$ stepped up, it has committed to stepping up and right until reaching $v_1$, as circling in place or going in a third direction increases its travel past $dist(A_1, v_1)$, causing the total travel of \textbf{ALG} to be greater than $T$--a contradiction.
    
    \item $v_1$ cannot be a vertex in the first column or in the top row except the top vertex of column $k$ or one vertex to its left, as should $v_1$ not equal those, settling there would block off the path to the top row going through the first column, and force other robots to travel to the top row through column $k$. This is sub-optimal, and causes the total travel to increase beyond $T$--a contradiction.
    
    \item $v_1$ cannot be any vertex in the $k$th column other than the top of the $k$th column, as this would require $A_1$ to step downwards.
    
\end{enumerate}

(*) From (1)-(3) we conclude that $v_1$ must equal precisely the top vertex of the $k$th column or one vertex to its left.

Up to the time when $A_1$ reaches the top row, none of the ends of the other columns have been seen, so \textbf{ALG} will run the same regardless of the value of $k$. Since total travel is assumed to be optimal, no robot can block $s$ for more than one time step, so by the time $A_1$ reaches the top row, there will have been created at least $4r$ robots. Each of these $4r$ robots must have already entered one of the columns or settled, since they travel optimal paths to their destination, and the total length of the bottom row is $20r^2$, whereas $30r^2$ time must have passed for $A_1$ to reach the top.

As there are $10r$ columns, there must exist a column that none of the robots $A_1, \ldots, A_{4r}$ have entered. Set the value of $k$ to equal this column.

When $A_1$ reaches $v_1$, the above indicates that any other robot currently present in the $k$th column (if there are any) arrived at least $2\cdot 4r$ time steps after $A_1$. Therefore it is at distance at least $8r$ from $A_1$, meaning that there is a space of $6r$ vertices in column $k$ that no robot has seen yet. This indicates that \textbf{ALG} must make the same decision for $A_1$ whether these vertices exist or not. However, if any one of these vertices does not exist, then column $k$ is not connected to the top row, indicating that $A_1$ cannot settle at the top of the $k$th column or to its left, else it will block off part of the environment. We arrived at a contradiction to (*).

We conclude that there is an environment $G(k)$ where the total travel of $\textbf{ALG}$ is greater than the optimum, so $ALG$ is sub-optimal. 

\end{proof}

By adding more columns to the $G(k)$ construction and increasing the height of the columns, we can force $A_1$ to go down more and more steps, causing the difference between the optimal total travel and the total travel of \textbf{ALG} to be arbitrarily large.

Proposition \ref{impossibility} only makes the assumption of limited visibility. It holds even assuming the agents have global communication, infinite memory, and are aware of each others' positions at all times.

We note that we did not exclude the possibility of a local rule that minimizes the maximal individual travel. Furthermore, we did not exclude the possibility of a rule that minimizes total travel when pauses are not counted. 

\section{Simulations, comparisons, and alternative strategies}

\begin{figure}[!ht]
  \centering%
    \includegraphics[width=.49\linewidth]{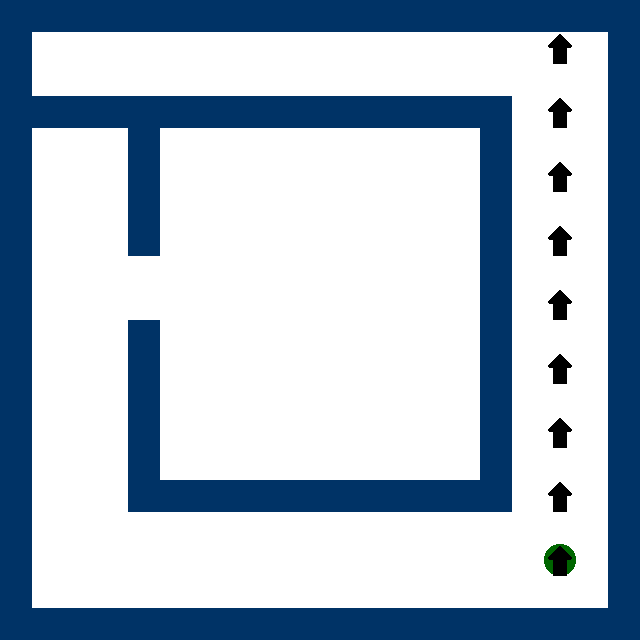}\hfill%
    \includegraphics[width=.49\linewidth]{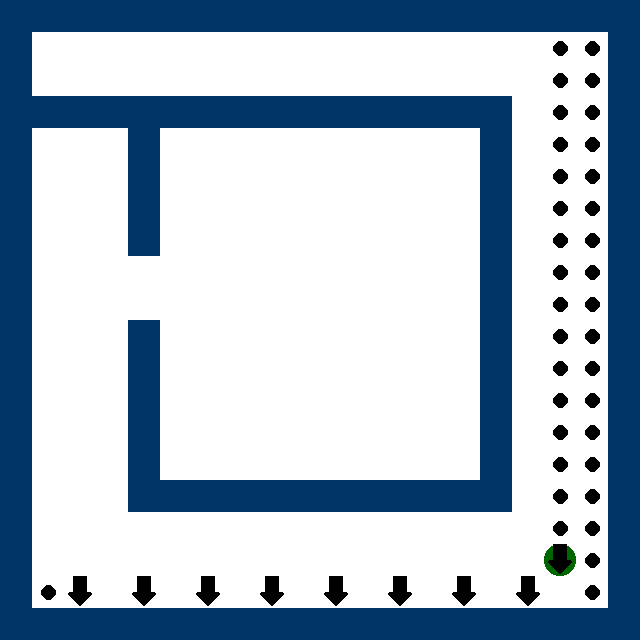}\hfill%
    \includegraphics[width=.49\linewidth]{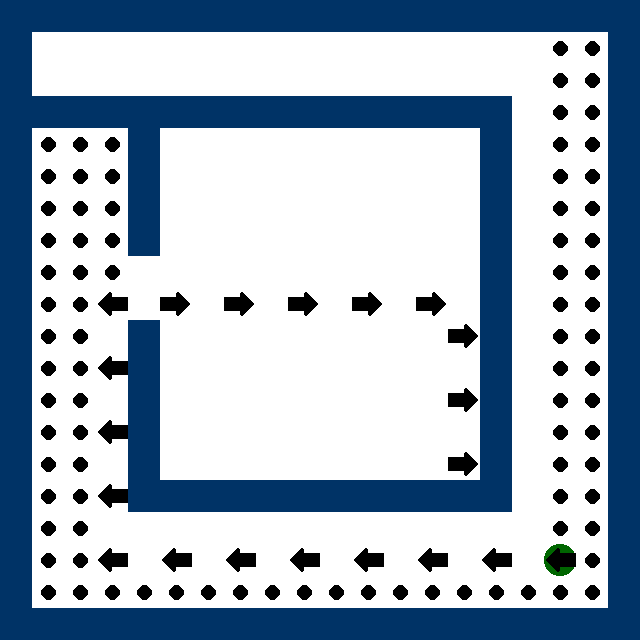}\hfill%
    \includegraphics[width=.49\linewidth]{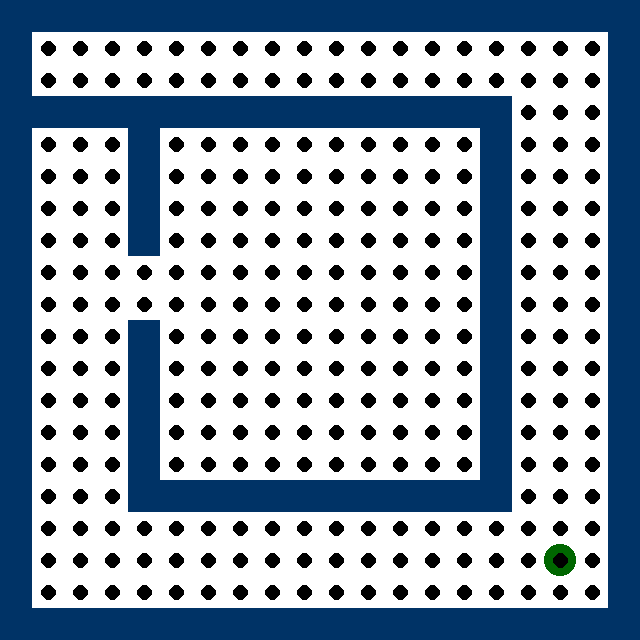}\hfill%
    \caption{A simulation of FCDFS. The blue blocks are walls. The arrows indicate the location and primary direction of the robots, and the diamonds are settled robots. 
    Rather than block active robots, the settled robots form halls to enable the swarm to explore more of the environment.}
    \label{fig:simulation2}
\end{figure}

\begin{figure}[!ht]
  \centering%
    \includegraphics[width=.49\linewidth]{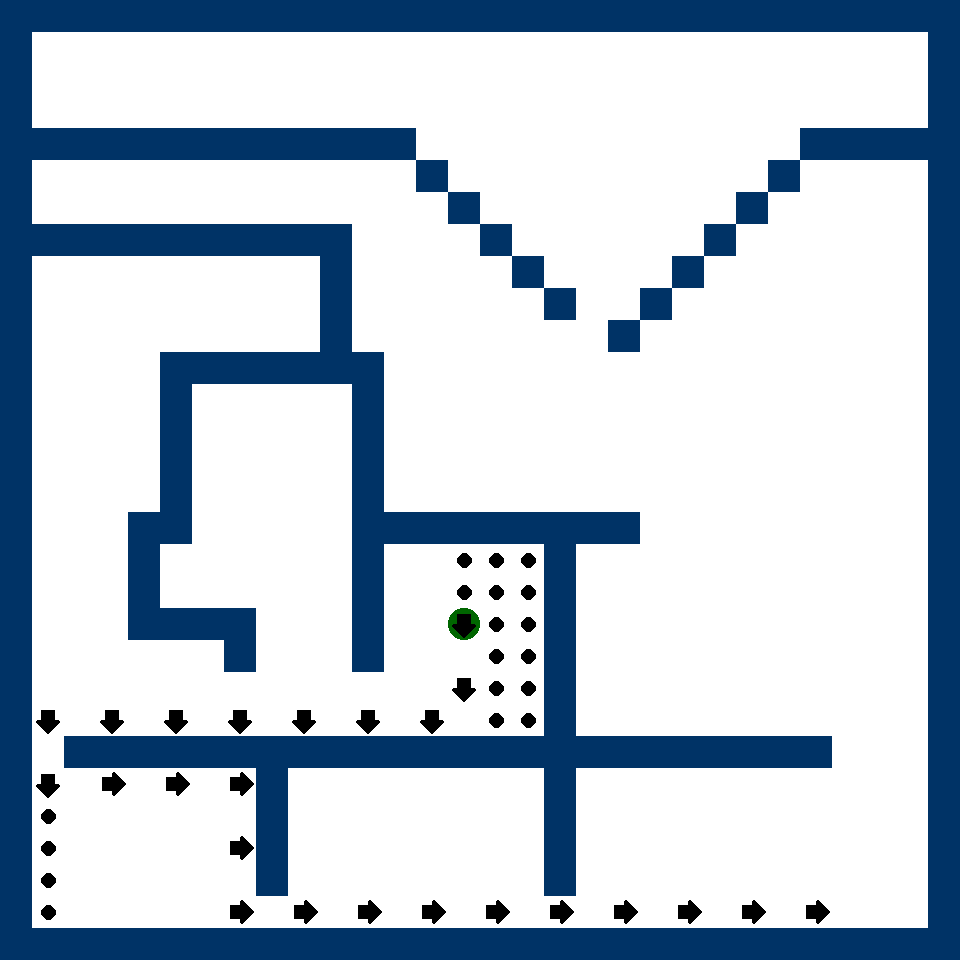}\hfill%
    \includegraphics[width=.49\linewidth]{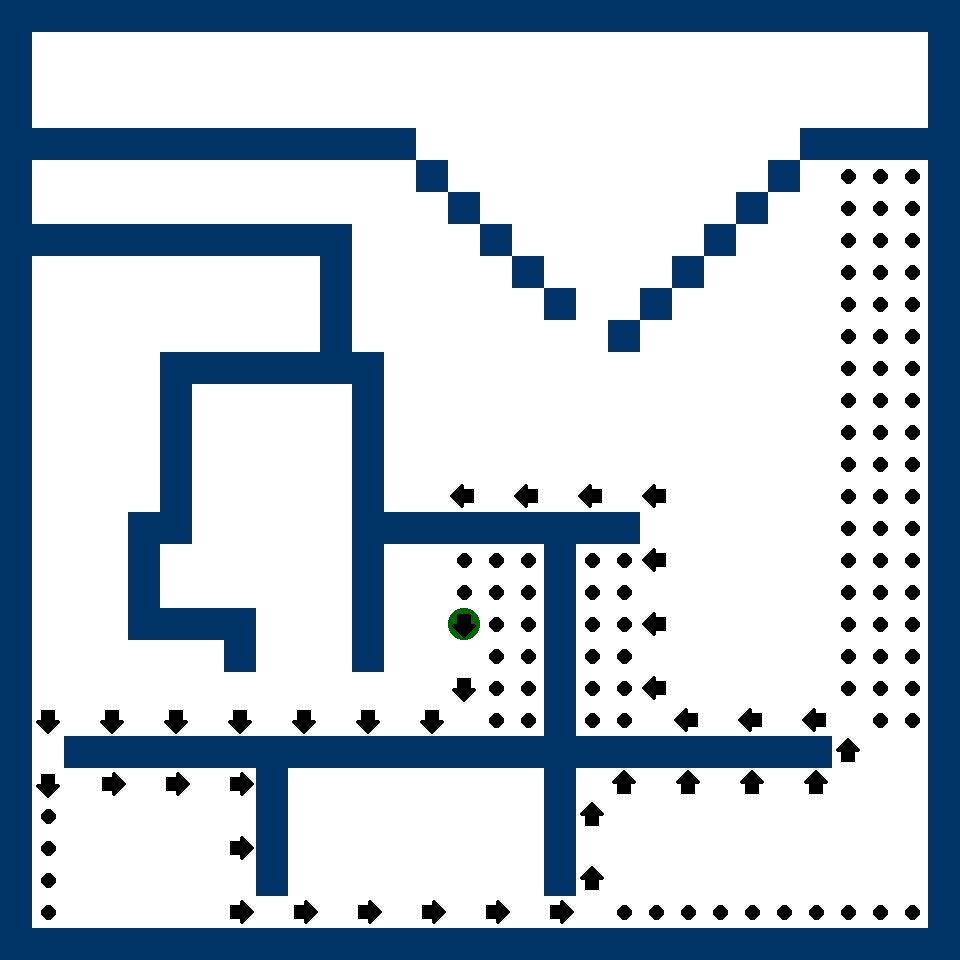}\hfill%
    \includegraphics[width=.49\linewidth]{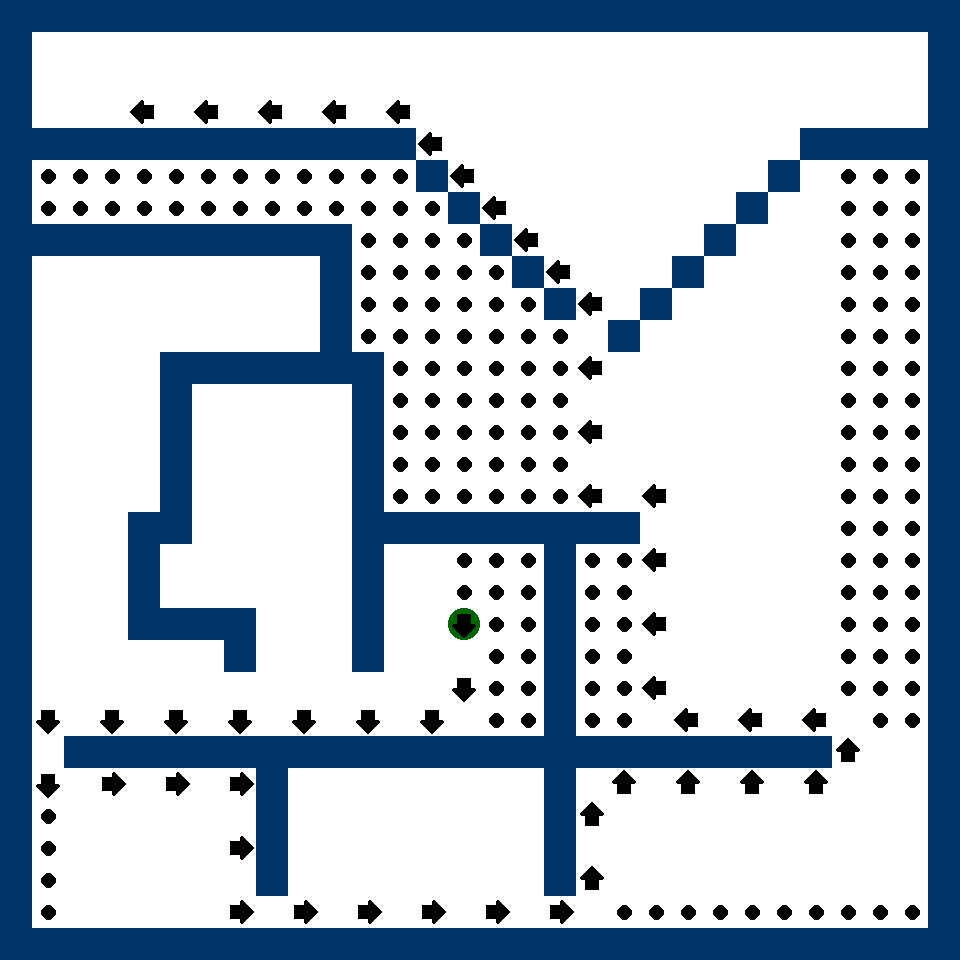}\hfill%
    \includegraphics[width=.49\linewidth]{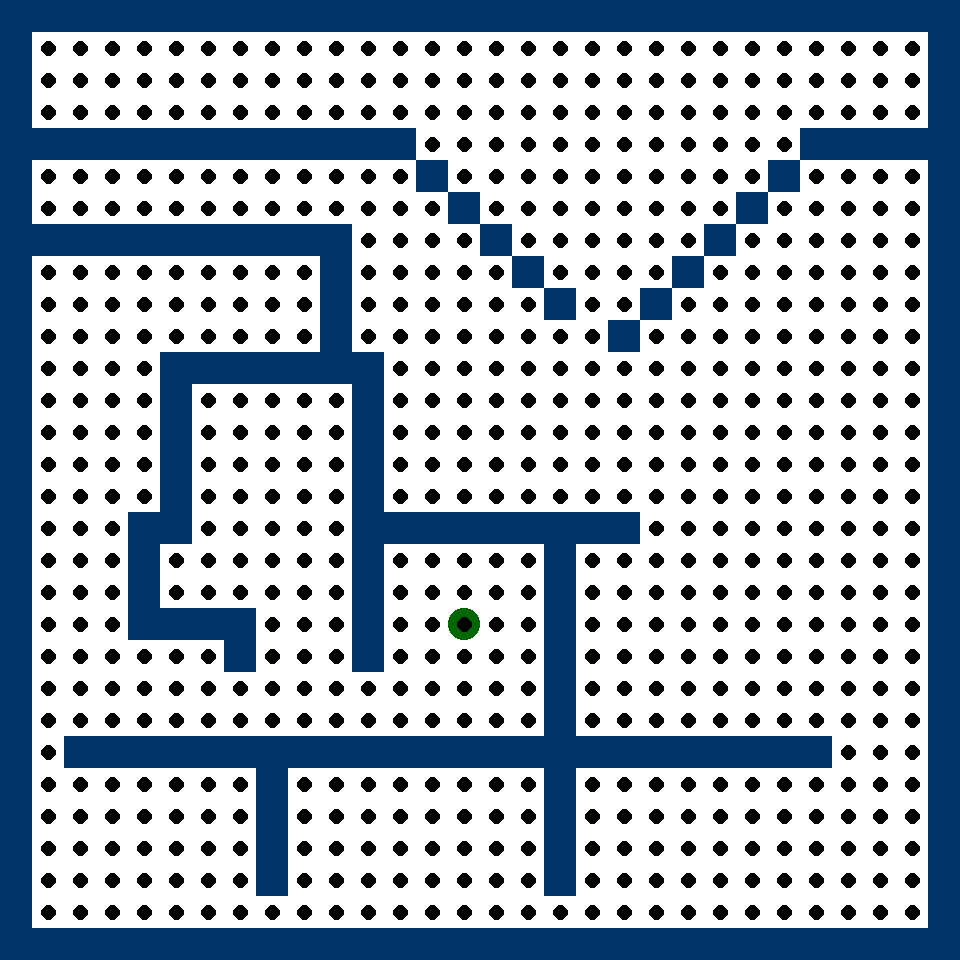}\hfill%
    \caption{A simulation on a different environment. Note how the trail of robots always forms a shortest path to its current front.}
    \label{fig:simulation}
\end{figure}

\begin{figure}[!ht]
    \centering
    \includegraphics[height=1.53in]{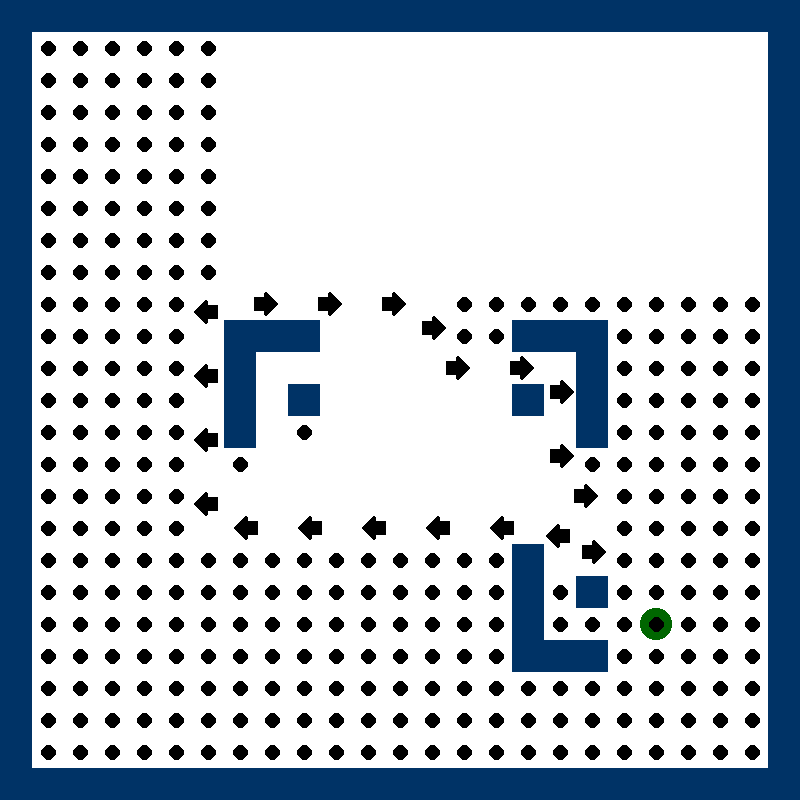}
    \caption{A deadlock scenario in environments that are not simply connected.}
    \label{fig:simulationdeadlock}
\end{figure}

\begin{figure}[htb]
    \includegraphics[height=1.53in]{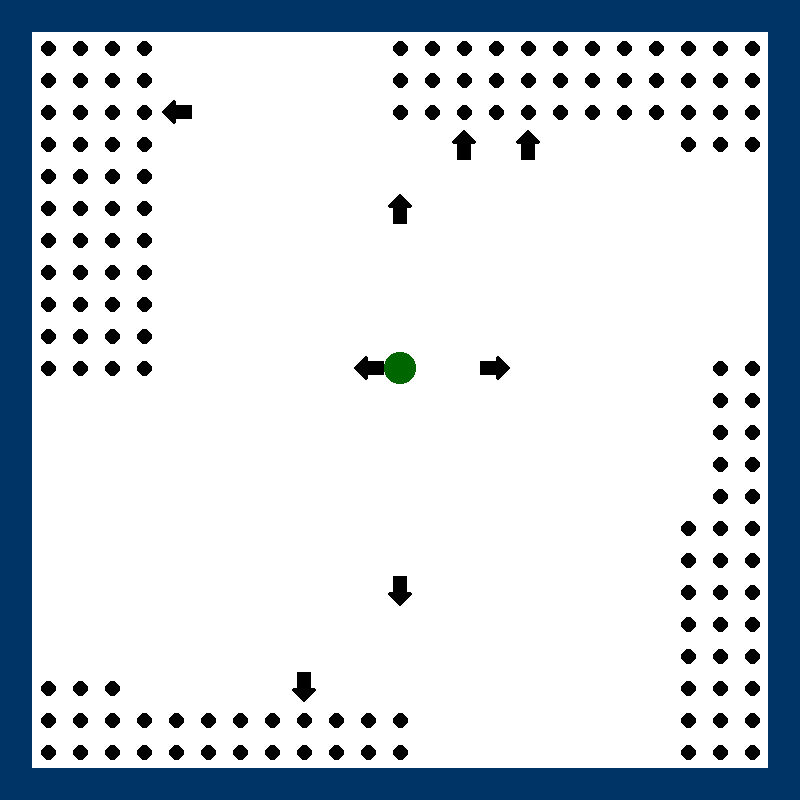}
    \caption{Multi-directional dispersal strategy.}
    \label{fig:alternatestrategies1} 
\end{figure}

\begin{figure}[htb]
    \includegraphics[height=1.53in]{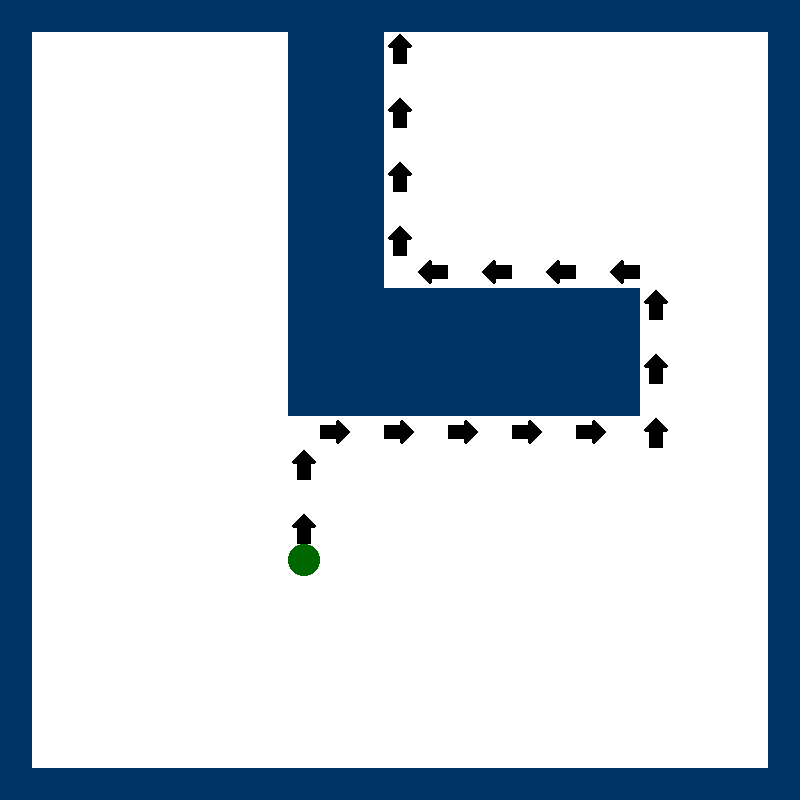}
    \caption{``Left hand on wall'' strategy.}
    \label{fig:alternatestrategies2} 
\end{figure}

We verified and animated our algorithm by simulating it on our robot simulator. Figures \ref{fig:simulation2} and \ref{fig:simulation} show four stills from a run of the algorithm on two different environments. Figure \ref{fig:simulationdeadlock} shows a FCDFS deadlock scenario in an environment that is not simply connected: the halls constantly redirect the robots, forming a cycle. The door vertex has mistakenly blocked itself off, due to the robots exiting from it mistaking the robots in a cycle for obstacles.

We experimented with two variants of FCDFS that are similarly optimal. FCDFS assumes robots are initialized with a common notion of up, down, left and right, but this assumption is unnecessary if we let robots settle in place as soon as they reach a corner (in FCDFS they keep moving if they can). This modified strategy is illustrated in Figure \ref{fig:alternatestrategies1}, where robots randomly choose their initial direction. This creates a more ``symmetric''-looking dispersal. The strategy shown in Figure \ref{fig:alternatestrategies2} is more significantly different: in it, rather than stick to their secondary and primary directions, robots attempt to scale the boundary of the environment with a ``left hand on wall'' clockwise orientation, until they hit a corner or a wall. Both of these variants achieved the same makespan and total travel as FCDFS, though they are visually distinct.

Empirically, we compared the performance of FCDFS to the performance of our implementation of the DFLF and BFLF algorithms of \cite{hsiang} (adapted to our slightly different model) over a number of simply-connected environments, measuring the total travel and maximal individual travel (Table \ref{tablecompare}). Note that though all algorithms are deterministic, some local decisions are not fully specified in \cite{hsiang}, hence different implementations may result in slightly different performance, though asymptotically every implementation will perform the same. We let our robots decide between arbitrary local decisions at random, averaging performance over several re-runs. 

Only for the sake of this comparison, we elected to exclude time steps where robots are active but do not change location, as such intermediate pauses are not counted in \cite{hsiang}. FCDFS is optimal regardless, and factoring these in leaves the DFLF and FCDFS columns unchanged, since such pauses never occur during their execution. However, including pauses causes the \textit{maximal travel} of BFLF to become extremely large. Hence, Table \ref{tablecompare} shows that BFLF is good at reducing the number of location changes of a robot, but in many applications (e.g. when robots are quadcopters) its \textit{energy} consumption is very high compared to FCDFS. 

\begin{table}[!h]
\begin{tabular}{@{}lllll@{} p{10cm}}
\toprule
                       & DFLF & \begin{tabular}[x]{@{}c@{}}BFLF\\(excl. stops)\end{tabular} & FCDFS &  \\ \midrule
30x30 Grid             &   237984 (460)   & 16323 (50)     &   13620 (32)   &  \\
{Fig. 5 Environment} &  16139 (126)    &  6742 (50)    &     5909 (38)   &  \\
{Fig. 6 Environment} &  {100419 (296)}   &   {39576 (112)}   &   {35103 (99)}    &  \\ 
{Fig. 9 Environment}             & {50889  (190)}   &   {7283 (39)}  &   {6600 (35)}    &  \\
\bottomrule
\end{tabular}
\newline
\newline
\caption{A comparison of total travel and maximal individual travel over different environments (excluding pauses). Entries are in the form \textit{total travel (maximal travel)}. See Figures 5, 6, 9 for the specific environments used.}
\label{tablecompare}
\end{table}

\section{Discussion}

A robotic swarm must take into account the  energy capacity of the individual. We discussed the problem of minimizing travel, hence energy expenditure, in the uniform dispersal problem for simply connected grid regions. We showed the existence of a strategy that minimizes total and individual travel for the case of a single source vertex. We showed also a non-existence result for such strategies in the case of general grid environments. 

Several extensions of our work can readily be considered. First, as our algorithm deals only with the single door case, it is desirable to find an energy-efficient dispersal algorithm for the case of multiple doors from which robots arrive independently.

Next, synchronicity is a strong assumption, enabling every robot to proceed to its destination without ever being blocked by another robot. To extend our work to less controlled settings, we may assume an asynchronous time scheme--for example, allow a probability $q$ that an agent fails to activate at a given time step. We cannot expect a makespan- and total travel-\textit{optimal} algorithm to exist in such settings, but we anticipate relatively effective strategies might exist. As a way to proceed, though our algorithm makes the powerful assumption of synchronicity, the strategy of finding corners and not stopping at halls seems general, and could possibly be adopted for the asynchronous case as well.

Finally, our algorithm requires the environments to be simply-connected orthogonal environments: it would be interesting to see an algorithm that works for broader scenarios, or in the opposite direction, results regarding the non-existence of efficient algorithms for such scenarios under the stringent computational assumptions we made.